\documentclass[a4paper,11pt]{article}
\usepackage{amsmath,amssymb,amsthm}
\usepackage{amsfonts}
\usepackage{eucal}
\numberwithin{equation}{section}

\renewcommand{\a}{\alpha}
\renewcommand{\b}{\beta}
\renewcommand{\c}{\gamma}
\renewcommand{\d}{\delta}
\newcommand{\e}{\varepsilon}
\newcommand{\f}{\varphi}

\newcommand{\m}{\mu}

\newcommand{\s}{\sigma}
\newcommand{\x}{\xi}
\newcommand{\y}{\eta}
\newcommand{\z}{\zeta}

\newcommand{\co}{\mathbb{C}}
\newcommand{\re}{\mathbb{R}}
\newcommand{\ze}{\mathbb{Z}}
\newcommand{\T}{\mathbb{T}}

\def\pa{\partial}

\newcommand{\supp}{\mathrm{{supp}}}

\DeclareMathOperator*{\esssup}{ess\ sup}

\newcommand{\bigpare}[1]{\bigl(#1\bigr)}

\newcommand{\bigbra}[1]{\bigl\{#1\bigr\}}

\newcommand{\biggbra}[1]{\biggl\{#1\biggr\}}

\newcommand{\bigset}[2]{\bigl\{#1\bigm|#2\bigr\}}

\newcommand{\norm}[1]{\| #1 \|}
\newcommand{\bignorm}[1]{\bigl\| #1 \bigr\|}

\newcommand{\bigabs}[1]{\bigl| #1 \bigr|}

\newcommand{\jap}[1]{\langle #1 \rangle}

\newtheorem{thm}{Theorem}[section]
\newtheorem{lem}[thm]{Lemma}

\newtheorem{cor}[thm]{Corollary}

\theoremstyle{definition}

\newtheorem{ass}{Assumption}

\theoremstyle{remark}

\title{On a continuum limit of discrete Schr\"odinger operators on square lattice}
    \author{Shu N{\sc akamura}\footnote{Graduate School of Mathematical Sciences, the University of Tokyo, 3-8-1 Komaba, Meguro, Tokyo, 153-8914, Japan, E-mail: {\tt shu@ms.u-tokyo.ac.jp}. 
Partially supported by JSPS Grant Kiban-B 15H03622.} 
and Yukihide T{\sc adano}\footnote{Graduate School of Mathematical Sciences, the University of Tokyo, 3-8-1 Komaba, Meguro, Tokyo, 153-8914, Japan, 
E-mail: {\tt tadano@ms.u-tokyo.ac.jp}. Supported by JSPS Research Fellowship for Young Scientists 17J05051.}}

\begin{document}

\maketitle

\begin{abstract}
The norm resolvent convergence of discrete Schr\"odinger operators to a continuum Schr\"odinger 
operator in the continuum limit is proved under relatively weak assumptions. 
This result implies, in particular, the convergence of the spectrum with respect to the Hausdorff distance. 
\end{abstract}

\section{Introduction}

We consider a Schr\"odinger operator 
\[
H = H_0+V(x), \quad H_0=-\triangle, \quad x\in\re^d, 
\]
on $\mathcal{H}=L^2(\re^d)$, where $d\geq 1$, and corresponding discrete Schr\"odinger operators: 
We set $h>0$ be the mesh size, and we write 
\[
\mathcal{H}_h=\ell^2(h\ze^d), \quad 
h \ze^d = \bigset{(hz_1,\dots,hz_d)}{z\in\ze^d}, 
\]
with the norm $\norm{v}_h^2= h^d \sum |v(hz)|^2$ for $v\in\mathcal{H}_h$. 
We denote the standard basis of $\re^d$ by $e_j=(\d_{ik})_{k=1}^d\in\re^d$, $j=1,\dots,d$.  
Our discrete Schr\"odinger operator is 
\[
H_h =H_{0,h} +V(z), \quad z\in h\ze^d, 
\]
where 
\[
H_{0,h} v(z)=h^{-2}\sum_{j=1}^d ( 2v(z) - v(z+he_j) - v(z-he_j)), \quad v\in\mathcal{H}_h. 
\]

We suppose 

\begin{ass}\label{ass-V}
$V$ is a real-valued continuous function on $\re^d$, and bounded from below. 
$(V(x)+M)^{-1}$ is uniformly continuous with some $M>0$, and there is $c_1>0$ such that 
\[
c_1^{-1} (V(x)+M) \leq V(y)+M\leq c_1(V(x)+M), \quad \text{if }|x-y|\leq 1. 
\]
\end{ass}

The above assumption implies $V$ is slowly varying in some sense, and uniformly continuous 
relative to the size of $V(x)$. Under the assumption, $H$ is essentially self-adjoint, and 
$H_h$ is self-adjoint. 
The assumption is satisfied if $V$ is bounded and uniformly continuous. 
$V(x)=a\jap{x}^\m$ with $a,\m>0$, also satisfies the assumption.

For $\f\in \mathcal{S}(\re^d)$, $h>0$ and $z\in h\ze^d$, we set 
\[
\f_{h,z}(x)= \f(h^{-1}(x-z)), \quad x\in \re^d, 
\]
and we define $P_h = P_{h,\f}: \mathcal{H} \to \mathcal{H}_h$ by
\[
P_h u(z) := h^{-d} \int_{\re^d} \overline{\f_{h,z}(x)} u(x) dx, \quad h > 0,\ z \in h\ze^d. 
\]
The adjoint operator is given by
\[
P_h^* v(x) = \sum_{z\in h\ze^d} \varphi_{h,z}(x) v(z) , \quad h > 0, \ v\in\mathcal{H}_h.
\]
It is easy to observe that $P_h^*$ is an isometry and hence $P_h$ is an orthogonal projection 
if and only if $\bigbra{\f_{1,z}\,|\,z\in\ze^d}$ is an orthonormal system. 
This condition is also equivalent to the condition:
\begin{equation}\label{eq-orthonormal-condition}
\sum_{n\in\ze^d} \bigabs{\hat \f(\x+n)}^2 =1\quad \text{for }\x\in\re^d, 
\end{equation}
where $\hat\f$ is the Fourier transform:
\[
\hat \f(\x) =\mathcal{F}\f(\x)= \int_{\re^d}e^{-2\pi i x\cdot\x}\f(x)dx, \quad \x\in\re^d. 
\]
This claim is well-known, but we give its proof in Appendix for the completeness (Lemma~\ref{lem-orthonormal}). 
By this observation, we learn that there is a large class of $\f$'s satisfying the above 
condition. In this paper, we use $P_h$ to identify $\mathcal{H}_h$ with a subspace of 
$\mathcal{H}$. We suppose:

\begin{ass}\label{ass-phi}
$\f$ satisfies the condition \eqref{eq-orthonormal-condition}, and 
$\supp[\hat\f]\subset (-1,1)^d$. 
\end{ass}

\begin{thm}\label{thm-main}
Suppose Assumptions~\ref{ass-V} and \ref{ass-phi}. Then, for any $\m\in \co \backslash \re$,
\[
\| P_h^*(H_h- \m)^{-1}P_h - (H-\m)^{-1}\|_{\mathcal{B}(\mathcal{H})} \to 0
\quad \text{as }h\to0.
\]
Furthermore, if $(V(x)+M)^{-1}$ is uniformly H\"older continuous of order $\a\in (0,1]$ (with some $M>0$), 
then for any $0<\b<\a$, 
\[
\| P_h^*(H_h- \m)^{-1}P_h - (H-\m)^{-1}\|_{\mathcal{B}(\mathcal{H})} \leq C_\m h^{\b}
\quad \text{as }h\to 0. 
\]
\end{thm}

Here $\mathcal{B}(X)$ denotes the Banach space of the operators on a Banach space $X$. 
Combining this with the argument of Theorem~VIII.23 (b) in \cite{R-S}, we obtain the 
following corollary. We denote the spectrum of a self-adjoint operator $A$ by $\s(A)$, 
and the spectral projection by $E_A(\Omega)$ for $\Omega\subset\re$.

\begin{cor}
Suppose Assumptions~\ref{ass-V} and \ref{ass-phi}. 
Let $a,b\in\re$, $a<b$, be not in $\s(H)$. Then $a,b \notin\s(H_h)$ for sufficiently small $h$ and
\[
\| P_h^* E_{H_h}((a,b)) P_h - E_H((a,b)) \|_{\mathcal{B}(\mathcal{H})} \to 0 \quad \text{as }h\to0.
\]
\end{cor}

We denote the Hausdorff distance of sets $X,Y\subset \co$ by 
\[
d_{\text{H}}(X,Y) = \max\biggbra{\sup_{x\in X}d(x,Y),\sup_{y\in Y}d(y,X) }, 
\]
where $d(\cdot,\cdot)$ denotes the standard distance in $\co$. 
It is not difficult to show $d_{\text{H}}(\s(A),\s(B))\leq \norm{A-B}$ for 
normal operators $A$ and $B$ (see Lemma~\ref{lem-hausdorff} in Appendix). 
Thus we also have the following result. 

\begin{cor} 
Suppose Assumptions~\ref{ass-V} and \ref{ass-phi}. Then for $M\gg 0$, 
\[
d_{\text{H}} \left( \s((H_h+M)^{-1}), \s((H+M)^{-1}) \right) \to 0  \quad \text{as }h\to0.
\]
\end{cor}

There are studies concerning continuum limits of NLS equations, in many cases, mainly with 
applications to numerical analysis. 
We refer Bambusi and Penati \cite{B-P}, Hong and Yang \cite{H-Y} and references therein.
For linear discrete Schr\"odinger operators, Rabinovich \cite{R} has studied the relation between the essential and discrete spectra of the discrete and continuum Schr\"odinger operators, provided $V$ is bounded and uniformly continuous.

In Section~2, we give the proof of our main theorem, and proofs of several technical lemmas are given in Appendix.


\section{Proof}\label{section-Proof}

We denote the discrete Fourier transform $F_h$ : $\mathcal{H}_h \to \hat{\mathcal{H}}_h =L^2(h^{-1}\T^d)$, 
$\T=\re/\ze$, by 
\[
F_h v(\z) = h^{d} \sum_{z\in h\ze^d} e^{-2\pi iz\cdot \z}v(z), \quad \z \in h^{-1}\T^d, \  v\in\mathcal{H}_h. 
\]
$F_h$ is unitary, and its adjoint is given by 
\[
F_h^*g(z) = \int_{h^{-1}\T^d} e^{2\pi i z\cdot\z} g(\z)d\z, \quad z\in h\ze^d, \  g\in\hat{\mathcal{H}}_h. 
\]

\subsection{Convergence of the free Hamiltonian} 

If we set $H_0(\x)= |2\pi \x|^2$, it is well-known that
$H_0 = \mathcal{F}^* H_0(\cdot)\mathcal{F}$ on $\mathcal{H}$. 
Similarly, if we set 
\[
H_{0,h}(\z) = 2 h^{-2} \sum_{j=1}^d  (1-\cos(2\pi h\z_j)), \quad \z\in h^{-1}\T^d, 
\]
then $H_{0,h} = F_h^*H_{0,h}(\cdot)F_h$. We denote
\[
Q_h := F_h P_h \mathcal{F}^*: \hat{\mathcal{H}} \to \hat{\mathcal{H}}_h.
\]
The following formula is convenient in the following argument. 
It is well-known in signal analysis (see, e.g., \cite{OS}), but we give a proof in Appendix for 
the completeness.

\begin{lem}\label{lem-Q} 
For $f\in\mathcal{S}(\re^d)$,  
\begin{equation}\label{eq-Q1}
Q_h f(\z) = \sum_{n\in\ze^d} \overline{\hat{\f}(h\z+n)} f(\z+ h^{-1}n), \quad \z\in h^{-1}\T. 
\end{equation}
For $g\in\hat{\mathcal{H}}_h$, 
\begin{equation}\label{eq-Q2}
Q_h^* g(\x) = \hat{\f}(h\x) \tilde g(\x), \quad \x\in\re^d,
\end{equation}
where $\tilde g$ is the periodic extension of $g$ on $\re^d$.
\end{lem}

\begin{lem}\label{lem-free-1}
For $\m\in\co\setminus\re_+$ there is $C>0$ such that
\[
\| (1-P_h^* P_h)(H_0-\m)^{-1} \|_{\mathcal{B}(\mathcal{H})} \leq C h^2, \quad h>0. 
\]
\end{lem}

\begin{proof}
We first note 
\[
\|(1-P_h^* P_h)(H_0-\m)^{-1}\|_{\mathcal{B}(\mathcal{H})} = \|(1-Q_h^* Q_h)(|2\pi \x|^2-\m)^{-1}\|_{\mathcal{B}(\mathcal{\hat{H}})},
\]
where $\hat{\mathcal{H}}=\mathcal{F}[\mathcal{H}]=L^2(\re^d)$. 
Let $f \in \hat{\mathcal{H}}$ and $g=(|2\pi \x|^2-\m)^{-1} f$. Then we have, by using the above lemma, 
\[
(1-Q_h^* Q_h)g(\x)
 = (1-|\hat{\varphi}(h\x)|^2)g(\x) 
 - \hat{\varphi}(h\x) \sum_{n \neq 0} \overline{\hat{\varphi}(h\x+n)} g(\x+h^{-1}n).
\]
For the first term in the right hand side, we observe by Assumption~\ref{ass-phi} that 
$|\hat{\varphi}(h\x)|=1$ if $|\x|\leq h^{-1}\d$ with some $\d>0$. Then we learn 
\[
\| (1-|\hat{\varphi}(h\x)|^2)g(\x) \|_{\hat{\mathcal{H}}} 
\leq \sup_{|\x| > h^{-1}\d} \left||2\pi \x|^2-\m\right|^{-1} \|f\|_{\hat{\mathcal{H}}} 
\leq Ch^2 \|f\|_{\hat{\mathcal{H}}} .
\]
For the second term, we note that the terms in the summation vanish except for $n\in \{0,\pm 1\}^d\setminus 0$. 
Using the support condition of $\hat{\f}$ again, we learn that 
$\hat\f(h\x) \overline{\hat\f(h\x+n)}=0$ if $|\x+h^{-1}n|\leq h^{-1}\d$ with some  $\d>0$. 
Thus we can use the same argument to show that the second term is bounded by $Ch^2$.
\end{proof}

\begin{lem}\label{lem-free-2}
For $\m\in\co\setminus\re_+$ there is $C>0$ such that
\[
\bignorm{(H_{0,h}-\m)^{-1}P_h - P_h(H_0-\m)^{-1}}_{\mathcal{B}(\mathcal{H},\mathcal{H}_h)} \leq Ch^2, 
\quad h>0. 
\]
\end{lem}

\begin{proof}
Since $P_h^*$ is isometric, it suffices to estimate
\begin{align*}
&\bignorm{(H_{0,h}-\m)^{-1}P_h - P_h(H_0-\m)^{-1}} \\
&\quad = \bignorm{P_h^*(H_{0,h}-\m)^{-1}P_h - P_h^*P_h(H_0-\m)^{-1}} \\
&\quad = \bignorm{Q_h^*(H_{0,h}(\cdot)-\m)^{-1}Q_h - Q_h^* Q_h(H_0(\cdot)-\m)^{-1}}. 
\end{align*}
Then we compute, for $f\in\mathcal{S}(\re^d)$, 
\begin{align*}
&\bigpare{ Q_h^*(H_{0,h}(\cdot)-\m)^{-1} Q_h - Q_h^*Q_h (H_0(\cdot)-\m)^{-1} } f(\x) \\
&\quad =\sum_{n\in\ze^d} \hat\f(h\x) \overline{\hat{\f}(h\x+n)}B_h(\x+h^{-1}n) f(\x+ h^{-1}n), 
\end{align*}
where $B_h(\x):= (H_{0,h}(\x)-\m)^{-1} - (H_0(\x)-\m)^{-1}$. 
We note, as well as in the proof of Lemma~\ref{lem-free-1}, $\hat\f(h\x) \overline{\hat{\f}(h\x-n)}$ 
vanishes except for $n\in \{0,\pm1\}^d$. 

By the Taylor expansion, we have 
\[
\bigabs{H_{0,h}(\x) -H_0(\x)}\leq Ch^{-2}(h|\x|)^4 = Ch^2|\x|^4, \quad h>0, \ \x\in\re^d. 
\]
On the other hand, if $h\x\in\supp[\hat\f]$, we have 
$H_{0,h}(\x)\geq c_0|\x|^2$  with some $c_0>0$. 
These imply 
\[
|\hat\f(h\x)|^2 |B_h(\z)| \leq C h^2 |\hat\f(h\x)|^2, \quad h>0, \ \x\in\re^d, 
\]
with some $C>0$. 
On the support of $\hat\f(h\x) \overline{\hat{\f}(h\x+n)}$, $n\neq 0$, 
we have $H_{0,h}(\x+h^{-1}n)\geq c_1 h^{-2}$, $H_0(\x+h^{-1}n)\geq c_1h^{-2}$ with some $c_1>0$, 
and hence $|B_h(\x)|=O(h^2)$ as $h\to 0$. 
Combining these, we learn 
\begin{align*}
&\bigabs{\bigpare{ Q_h^*(H_{0,h}(\cdot)-\m)^{-1} Q_h - Q_h^*Q_h (H_0(\cdot)-\m)^{-1} } f(\x)}\\
&\quad \leq Ch^2 \sum_{n\in\{0,\pm1\}^d} |f(\x+h^{-1}n)|, \quad \x\in\re^d, 
\end{align*}
and the assertion follows. 
\end{proof}

\subsection{Relative boundedness}

In this section, we suppose $V\geq 1$ without loss of generality. 
In particular, $V(x)^{-1}$ is uniformly bounded, and 
\begin{equation}\label{eq-cond-V-2}
c_1^{-1} V(x) \leq V(y)\leq c_1 V(x) \quad \text{for } x,y\in\re^d,\ |x-y|\leq 1.
\end{equation}

\begin{lem}\label{lem-rel-bound-1}
Suppose Assumption~\ref{ass-V}. Then $V$ is $H$-bounded, and hence $H_0$ is also $H$-bounded. 
\end{lem}

\begin{proof}
By the quadratic inequality, it is easy to observe  $V^{1/2}$ and $(H_0+1)^{1/2}$ are $H^{1/2}$-bounded. 
Let $\y\in C_0^\infty(\re^d)$ be a smooth cut-off function such that $\y(x)\geq 0$, 
$\supp[\y]\subset\{|x|\leq 1\}$ and $\int \y(x)dx =1$. 
Then we set $\tilde V = \y*V$, and we use $\tilde V\geq 1$ as a smooth weight function comparable to $V$. 
By \eqref{eq-cond-V-2}, we have 
\[
c_1^{-1} V(x)\leq \tilde V(x)\leq c_1 V(x), \quad x\in\re^d.
\]
By elementary computation, we also have 
\[
\bigabs{\pa_x^\a \tilde V(x) }\leq C_\a \tilde V(x), \quad x\in\re^d
\]
with some $C_\a>0$, where $\a\in\ze_+^d$. It suffices to show $\tilde V$ is $H$-bounded.  

We write $W(x)=\tilde V(x)^{1/2}\geq 1$, and compute 
\begin{align*}
\tilde VH^{-1} &= WH^{-1}W + W[W,H^{-1}] \\
&= (WH^{-1/2})(WH^{-1/2})^* + W H^{-1}[H,W]H^{-1}.
\end{align*}
The first term in the right hand side is bounded since $W$ is $H^{1/2}$-bounded. 
We note 
\[
[H,W] = -\pa_x\cdot \pa_x W(x)- \pa_x W(x)\cdot \pa_x,
\]
and  $\pa_x$ is $H^{1/2}$-bounded. We also note 
\[
\bigabs{\pa_x W(x)} = \tfrac12 \tilde V^{-1/2}(x) \bigabs{\pa_x \tilde V(x)} \leq CW(x)
\]
with some $C>0$, and hence $\pa_x W$ is $H^{1/2}$-bounded.  Thus we learn 
\begin{align*}
W H^{-1}[H,W]H^{-1} &= (WH^{-1/2})(\pa_x H^{-1/2})^* ((\pa_x W) H^{-1/2})H^{-1/2} \\
&\quad - (WH^{-1/2})((\pa_x W)H^{-1/2} )^*(\pa_x H^{-1/2})H^{-1/2}
\end{align*}
is bounded, and hence $\tilde V$ is $H$-bounded. 
\end{proof}

\begin{lem}\label{lem-rel-bound-2}
Suppose Assumption~\ref{ass-V}. Then $V$ is $H_h$-bounded uniformly in $h>0$, 
and hence $H_{0,h}$ is also $H_h$-bounded uniformly in $h>0$. 
\end{lem}

\begin{proof}
The proof is analogous to that of Lemma~\ref{lem-rel-bound-1}. 
We note $W=\tilde V^{1/2}$ and $H_{0,h}^{1/2}$ are uniformly $H_h^{1/2}$-bounded. 
We similarly have 
\[
\tilde V H_h^{-1}
=(WH_h^{-1/2})(WH_h^{-1/2})^* + W H_h^{-1}[H_h,W]H_h^{-1},
\]
and the first term in the right hand side is uniformly bounded. 

For the second term, we recall that $H_{0,h} = \sum_{j=1}^d \nabla_j^* \nabla_j$, where
\[
\nabla_j v(z) := \frac{1}{h}\left( v(z+he_j) - v(z) \right), \ v\in\mathcal{H}_h.
\]
Then we learn
\[
[W, H_h] = \sum_{j=1}^d  \bigpare{ [\nabla_j,W]^* \nabla_j - \nabla_j^* [\nabla_j,W] }. 
\]
By elementary computations, we can show $[\nabla_j,W] W^{-1}$ is bounded uniformly in $h$, and 
hence $W H_h^{-1}[H_h,W]H_h^{-1}$ is bounded uniformly in $h$. 
\end{proof}

\subsection{Proof of Theorem \ref{thm-main}}

\begin{lem}\label{lem-V}
If $G$ is a uniformly continuous function, then 
\[
\bignorm{GP_h - P_h G}_{\mathcal{B}(\mathcal{H},\mathcal{H}_h)} \to 0 , \quad h\to0.
\]
If, in addition, $G$ is uniformly H\"older continuous of order $\a\in(0,1]$, then 
\[
\bignorm{G P_h - P_h G}_{\mathcal{B}(\mathcal{H},\mathcal{H}_h)} \leq C_\e h^{\a-\e}, 
\quad h>0, 
\]
with any $\e>0$.
\end{lem}

\begin{proof}
We note 
\[
\bigpare{G P_h - P_h G} u(z) 
= \int_{\re^d} K(x,z;h) u(x) dx,
\]
where
\[
K(x,z;h):=h^{-d} (G(z) - G(x)) \overline{\varphi(h^{-1}(x-z))} .
\]
By Schur's lemma, we have 
\[
\bignorm{G P_h - P_h G}\leq \sqrt{K_1K_2}, 
\]
where 
\[
K_1= \sup_{z\in h\ze^d} \int_{\re^d} |K(x,z)| dx , \quad
K_2 =\esssup_{x\in\re^d} h^d \sum_{z\in h\ze^d} |K(x,z)| .
\]
We set
\[
R(\d) := \sup_{x,y\in\re^d,|x-y|<\d} |G(x)-G(y)|
\]
and we choose $n>d$. Then we have 
\begin{align*}
\int_{\re^d} |K(x,z)| dx 
&= \int_{|x-z|<\d} |K(x,z)| dx + \int_{|x-z|\geq\d} |K(x,z)| dx \\
& \leq C R(\d) \int_{|y|<\d} \jap{hy}^{-n} h^{-d} dy + C \int_{|y|\geq \d} \jap{hy}^{-n}h^{-d} dy \\
& \leq C'  R(\d) + C' \jap{h^{-1}\d}^{-(n-d)} .
\end{align*}
By the same computation, we also have 
\[
h^d \sum_{z\in h\ze^d} |K(x,z)| \leq  C R(\d) + C \jap{h^{-1}\d}^{-(n-d)}.
\]
Combining these and setting $\d=h^\c$ with $\c\in(0,1)$, we obtain 
\[
\bignorm{G P_h - P_h G}\leq C R(h^\c) + C h^{(1-\c)(n-d)}.
\]
By the assumption, $R(\d)\to0$ as $\d\to 0$, and we conclude the first assertion.
 
If $G$ is uniformly H\"older continuous of order $\a$, then $R(\d) \leq C \d^\a$, and hence 
the right hand side of the above estimate is $O(h^{\a\c})+O(h^{(1-\c)(n-d)})$. 
We can choose $\c$ very close to 1, and $n$ very large so that $\a\c\geq \a-\e$ and 
$(1-\c)(n-d)\geq \a-\e$, and we have the second assertion. 
\end{proof}

\begin{proof}[Proof of Theorem~\ref{thm-main}]
We compute
\begin{align*}
&P_h^*(H_h- \m)^{-1}P_h - (H-\m)^{-1} \\
&= P_h^*(H_h- \m)^{-1}P_h - P_h^* P_h (H-\m)^{-1} - (1-P_h^* P_h) (H-\m)^{-1} \\
&= P_h^* (H_h-\m)^{-1}(P_h H - H_h P_h)(H-\m)^{-1} - (1-P_h^* P_h) (H-\m)^{-1}. 
\end{align*}
By Lemmas~\ref{lem-free-1} and \ref{lem-rel-bound-1}, we learn 
\[
\bignorm{(1-P_h^* P_h) (H-\m)^{-1}}\leq Ch^2.
\]
The other term is estimated as follows: 
\begin{align*}
&\bignorm{(H_h-\m)^{-1}(P_h H - H_h P_h)(H-\m)^{-1}}\\
& \leq \bignorm{(H_h-\m)^{-1}(P_h H_0 - H_{0,h} P_h)(H-\m)^{-1}} \\
& \qquad  +\bignorm{(H_h-\m)^{-1}(P_h V - V_h P_h)(H-\m)^{-1}}\\
& \leq C\bignorm{(H_{0,h}-\m)^{-1}(P_h H_0 - H_{0,h} P_h)(H_0-\m)^{-1}} \\
& \qquad  +C\bignorm{(V-\m)^{-1}(P_h V - V P_h)(V-\m)^{-1}}\\
& = C\bignorm{(H_{0,h}-\m)^{-1}P_h - P_h(H_0-\m)^{-1}} \\
 & \qquad  +C\bignorm{(V-\m)^{-1}P_h - P_h(V-\m)^{-1}},
\end{align*}
where we have used Lemmas~\ref{lem-rel-bound-1} and \ref{lem-rel-bound-2} for the 
second inequality. 
The two terms in the right hand side are estimated using Lemmas~\ref{lem-free-2} and \ref{lem-V}, 
respectively, to complete the proof. 
\end{proof}


\appendix
\section{Appendix}

Here we give the proofs of several technical lemmas.

\begin{lem}\label{lem-orthonormal}
Let $\f\in \mathcal{S}(\re^d)$. Then, the following are equivalent.
\begin{enumerate}
\renewcommand{\labelenumi}{{\rm (\arabic{enumi})}}
\item $P_h^*$ is isometric. 
\item $\operatorname{Ran}P_h = \mathcal{H}_h$. 
\item $\int_{\re^d} \varphi(x) \overline{\varphi(x-n)} dx = \d_{n,0}$ for $n\in\ze^d$.
\item $\sum_{n\in\ze^d} |\hat{\varphi}(\xi+n) |^2 = 1$ for $\xi\in\re^d$, where $\hat \f =\mathcal{F}\f$. 
\end{enumerate}
\end{lem}

\begin{proof}
(1) and (2) are equivalent by the standard properties of adjoint operators.
Since (2) implies the orthonormality of the basis $\{ h^{-\frac{d}{2}}\f_{h,z} \}_{z\in h\ze^d}$, we learn
\begin{align*}
\int_{\re^d} \f(x) \overline{\f(x-n)} dx
 =  h^d \int_{\re^d} \f_{h,0}(x) \overline{\f_{h,hn}(x)} dx = \d_{0,n},
\end{align*}
which implies (3). For the equivalence of (3) and (4), we learn by Parseval's identity
\begin{align*}
\int_{\re^d} \f(x) \overline{\f(x-n)} dx
 &= \int_{\re^d} \hat{\f}(\x) \overline{e^{-2\pi in\cdot\x}\hat{\f}(\x)} d\x \\
 &= \int_{\re^d} e^{2\pi in\cdot\x} |\hat{\f}(\x)|^2 d\x \\
 &=\int_{\T^d} \sum_{m\in\ze^d} e^{2\pi in\cdot(\x+m)} |\hat{\f}(\x+m)|^2 d\x \\
 &= \int_{\T^d} e^{2\pi in\cdot\x} \sum_{m\in\ze^d} |\hat{\f}(\x+m)|^2 d\x,
\end{align*}
where $\T^d=(\re/\ze)^d\simeq [0,1)^d$. Since $\{ e^{2\pi in\cdot\x} \}_{n\in\ze^d}$ is a complete 
orthonormal basis of $L^2(\T^d)$, we conclude that (3) is equivalent to (4).
\end{proof}

\begin{lem}\label{lem-hausdorff}
For normal operators $A$ and $B$, $d_{\text{H}}(\s(A),\s(B))\leq \norm{A-B}$. 
\end{lem}

\begin{proof}
It suffices to show that $d(\m,\s(B)) > \| A-B \|$ implies $\m\notin\s(A)$.
This condition implies $\|(A-B)(B-\m)^{-1}\| <1$ and hence the Neumann series
\begin{align*}
&(A-\m)^{-1}=(B-\m+A-B)^{-1} \\
 &= (B-\m)^{-1}(1+ (A-B)(B-\m)^{-1})^{-1} \\
 &= (B-\m)^{-1}\sum_{n=0}^\infty (-1)^n \left( (A-B)(B-\m)^{-1} \right)^n
\end{align*}
converges, and thus we learn $\m\notin\s(A)$. 
\end{proof}

\begin{proof}[Proof of Lemma~\ref{lem-Q}] We compute 
\begin{align*}
Q_h f(\z) =&  h^d \sum_{z\in h\ze^d} e^{-2\pi iz\cdot \z}
\left( h^{-d} \int_{\re^d} \overline{\f_{h,z}(x)}  \int_{\re^d} e^{2\pi ix\cdot\x} f(\x) d\x dx \right)  \\
=& \sum_{z\in h\ze^d} e^{-2\pi iz\cdot \z}
\int_{\re^d} \overline{\f_{h,z}(x)} \int_{\re^d} e^{2\pi ix\cdot\x} f(\x) d\x dx  \\
=& h^d \sum_{z\in h\ze^d} \int_{\re^d} e^{2\pi iz\cdot(\x-\z)}  \overline{\hat{\f}(h\x)} f(\x) d\x  \\
=& h^d \sum_{z\in h\ze^d} \sum_{n\in\ze^d} \int_{h^{-1}(\T^d +  n)} e^{2\pi iz\cdot(\x-\z)}  \overline{\hat{\f}(h\x)} f(\x) d\x  \\
=&h^d \sum_{z\in h\ze^d} \sum_{n\in\ze^d} \int_{h^{-1}\T^d} e^{2\pi iz\cdot(\x-\z)}  \overline{\hat{\f}(h\x+ n)} f(\x+ h^{-1}n) d\x  \\
=& h^d \sum_{z\in h\ze^d}  \int_{h^{-1}\T^d} e^{2\pi iz\cdot(\x-\z)}  \sum_{n\in\ze^d} \overline{\hat{\f}(h\x+ n)} f(\x+ h^{-1}n) d\x  \\
=&\sum_{n\in\ze^d} \overline{\hat{\f}(h\z+ n)} f(\z+ h^{-1}n)   .
\end{align*}
We have used the Fourier inversion formula for the last equality. We also have 
\begin{align*}
\jap{Q_h^*g,f} &= \int_{h^{-1}\T^d} \sum_{n\in\ze^d} g(\z) \hat\f(h(\z+h^{-1}n)) \overline{f(\z+h^{-1}n)} d\z \\
&=\int_{\re^d} \tilde g(\x) \hat\f(h\x)\overline{f(\x)} d\x,
\end{align*}
and this implies \eqref{eq-Q2}. 
\end{proof}

\end{document}